\RequirePackage{fix-cm}
\documentclass[smallextended]{svjour3}       
\smartqed  
\usepackage{verbatim,graphicx}
\pdfoutput=1

\usepackage{amsfonts, amsmath, amssymb, latexsym}
\usepackage{algorithmic}
\usepackage{algorithm,color}

\begin{document}

\title{Fuchsian Codes for AWGN Channels}

\author{Iv\'an Blanco-Chac\'{o}n \and Dionis Rem\'{o}n \and  Camilla Hollanti \and Montserrat Alsina}


\institute{I. Blanco-Chac\'{o}n \at
              Aalto University, Department of Mathematics and Systems Analysis, P.O. Box 11100, FI-00076 AALTO, Helsinki, Finland.\\
              Partially supported by the Academy of Finland grant 131745. \\
              \email{ivan.blancochacon@aalto.fi}
           \and
           C. Hollanti \at
              Aalto University, Department of Mathematics and Systems Analysis, P.O. Box 11100, FI-00076 AALTO, Helsinki, Finland.\\
              Partially supported by the Academy of Finland grant 131745. \\
              \email{camilla.hollanti@aalto.fi}
           \and
           D. Rem\'{o}n \at
              University of Barcelona, Faculty of Mathematics. Gran Via de les Corts Catalanes 585, 08007 Barcelona (Spain). \\
              Partially supported by MTM2012-33830. \\
              \email{dremon@ub.edu}
           \and
           M. Alsina \at
              Universitat Polit\`{e}cnica de Catalunya- BarcelonaTech, Dept. Applied Mathematics III - EPSEM,
              Av. Bases de Manresa 61-73,  08242  Manresa (Spain)\\
              Partially supported by MTM2009-07024. \\
              \email{montserrat.alsina@upc.edu}
}

\date{Received: date / Accepted: date}

\maketitle

\begin{abstract}
We develop a new transmission scheme for additive white Gaussian noisy (AWGN) single-input single-output (SISO) channels  without fading based on arithmetic Fuchsian groups. The properly discontinuous character of the action of these groups on the upper half-plane translates into fast decodability.
\end{abstract}

\section*{Introduction}
\label{intro}
In the last ten years, a group of people have used Number Theory as a tool to construct powerful transmission schemes for fast fading multiple-input multiple-output (MIMO) channels. We refer to the use of the structure of orders in number fields and in central simple division algebras over number fields (cf. \cite{viterbo2} and \cite{viterbo1}).

In the setting of space-time block codes, in \cite{hollanti3} it is shown how to use class field theory to derive an upper bound for the coding gain of a space-time block code attached to a cyclic division algebra, and explicit constructions as well as simulations are also carried out. Space-time block codes obtained from cyclic division algebras via the left regular representation over a suitable center have non vanishing minimal determinants (NVD property), and this translates into better coding gains. Maximal orders in the context of space-time coding were first considered in \cite{hollanti4}.

In the present paper we also use some ideas from Number Theory to construct a transmission scheme. We assume one transmit (Tx) antenna, one receive (Rx) antenna, and an alphabet consisting of a finite collection of $4$-tuples of integers $\{(x_i,y_i,z_i,t_i)\}_{i=1}^{|C|}$, where $C$ is the codebook and $|C|< \infty$ its size. We will encode each $4$-tuple into a codeword to be transmitted by a single transmit antenna, and then reconstruct (decode) the sent $4$-tuple. Intuitively, we are sending $4$ integers simultaneously by a single antenna. However, our $4$-tuple will consist of 3 independent integers, and the fourth will be determined according to these as an additional check-up symbol.

To this end, we will suppose that each $4$-tuple determines a matrix belonging to an arithmetic Fuchsian group of the first kind attached to the maximal order (up to conjugation) of an indefinite quaternion $\mathbb{Q}$-algebra which has been fixed at the beginning. This condition translates into the fact that each $4$-tuple $(x,y,z,t)$ of our alphabet satisfies an equation of the form $x^2-ay^2-cz^2+abt^2=1$ for fixed $a,b\in\mathbb{Z}$. Geometrically, this means that our alphabet is contained in a $4$-dimensional hyper quadric.

Arithmetic Fuchsian groups of the first kind arise in the study of Shimura curves (cf. \cite{shimura1967}), a rich theory with a large number of theoretical applications in very deep branches of Number Theory, like in Jacquet-Langlands correspondence, the theory of canonical models or the proof of the Shimura-Taniyama-Weil conjecture (the main argument in Fermat's Last Theorem). But Shimura curves are also present in the theory of error-correcting codes (cf. \cite{elkies}). In addition, we mention that arithmetic Fuchsian groups have been used in the design of space-time block codes (cf. \cite{brasil1}). Our approach is of completely different nature, since our code is non-linear and has logarithmic complexity. Besides that, it uses in a crucial way a recent result: the point-reduction algorithm \ref{algoritmo1} and its generalization given in theorem \ref{redpointalg}. In particular, our transmission scheme can be thought of as a non-linear information-compressing scheme. We also remark that the usual reason as to why linear codes are preferred over non-linear ones, namely their simple linear decoding methods, do not apply here, since we shall also demonstrate how to decode the constructed non-linear code with lesser complexity than that of typical linear maximum-likelihood (ML) decoders, as we will soon explain.

Fixed an indefinite quaternion $\mathbb{Q}$-algebra, one can embed the group of units of reduced norm 1 in its maximal order into $\mathrm{M}_2(K)$ via the left-regular representation (for some totally real number field $K$), obtaining in this way a discrete group $\Gamma$ (cf. \cite{alsinabayer}).

The crucial properties which yield a remarkably efficient decoding are that the action of $\Gamma$ on the complex upper half plane (which throughout this paper will be denoted by $\mathcal{H}$) is properly discontinuous, and that there exists a recently produced  explicit algorithm by Rem\'on et al. (\cite{bayerremon}) to reduce points in $\mathcal{H}$ to any given fundamental domain, once a presentation of the group has been obtained. In fact, our algorithm does $O(\log (n))$ operations, where $n$ is the number of generating matrices of a given matrix input. We present in this paper a very particular version of this algorithm for a concrete arithmetic Fuchsian group, leaving the general case for forthcoming publications, since it would significantly increase the length of the present paper.

The paper is organized as follows: in Section 1 we describe in detail our transmission problem and give some coding theoretical motivations for our work. In Section 2, we review some terminology and facts about orders in quaternion algebras and arithmetic Fuchsian groups which will be used later. In Section 3 we develop the transmission scheme and explain the algorithm on which it is based. We compute the complexity of this algorithm. In Section 4 we apply the reduction point algorithm to the AWGN channel, and study how to generate the constellation. In Section 5 we plot the  error performance of our Fuchsian codes with constellation sizes $4, 8$ and $16$ by computing the codeword error rate (CER) as a function  of  signal-to-noise ratio (SNR). Finally, in Section 6 we discuss some topics for further research.

\section{Generalities}
\label{sec:1}
We are interested in sending $4$-tuples of integers $(x,y,z,t)$ subject to the restriction $x^2-ay^2-bz^2+abt^2=1$, where $a>0$ and $b<0$ and each component can take a finite prescribed number of values. We will send this $4$-tuple as the complex signal $\gamma(\tau)$ where $\gamma=\left(\begin{array}{cc}x+\sqrt{a}y & z+\sqrt{a}t\\b(z-\sqrt{a}t) & x-\sqrt{a}y\end{array}\right)$ and $\tau$ is an element in the complex upper half-plane $\mathcal{H}$ to be determined for optimality. We are interested in the case when the matrices belong to a finite subset $G$ of an arithmetic Fuchsian group $\Gamma$. In this case, we will refer to the set $\{\gamma(\tau)|\gamma\in G\}$ as a Fuchsian code. We will take $\tau$ such that $\gamma_1(\tau)\neq\gamma_2(\tau)$ if $\gamma_1\neq\pm\gamma_2$. To obtain this, it is enough that $\tau$ is an interior point of a fundamental domain for $\Gamma$ acting on $\mathcal{H}$. Hence, in this case, we have as many codewords as $4$-tuples. We will suppose that our channel is affected by additive white Gaussian noise (AWGN), having a channel equation
$$
v=u+n,
$$
where $u=\gamma(\tau)$ is the transmitted signal, $v$ the received signal and $n\in\mathbb{C}$ is a complex circular Gaussian random variable.

\begin{definition}Let $C$ be a code. The \emph{data rate} in bits per channel use (bpcu) of $C$ is defined as
$
R=\dfrac{\log_2(|C|)}{N},
$
where $N$ is the number of channel uses. Notice that in our case $N=1$.

The \emph{code rate} of $C$ in (real) dimensions per channel use (dpcu) is defined as
$
R_c=\dfrac{\dim_\mathbb{Q}(C)}N,
$
where $\dim_\mathbb{Q}(C)$ is the number of independent integer symbols in the codeword and $N$ is again the number of channel uses. Notice that in our case the code rate is 3 dpcu.

\end{definition}

We will refer to the set of codewords as non uniform Fuchsian (NUF) constellation. It is important to point out that, due to the algebraic dependence of the $4$ integers in each of our $4$-tuples, we are essentially sending three independent integers at a time by a unique antenna. We will carry out simulations for $4$-NUF ($2$ bpcu), $8$-NUF ($3$ bpcu) and $16$-NUF ($4$ bpcu) constellations.

In this paper by signal-to-noise ratio we will mean the quotient $SNR=10\log_{10}\left(\dfrac{E}{N_0}\right)$ (dB), where $E$ is the average energy, \emph{i.e.}, $E=\frac{1}{|C|}\sum_{k=1}^{|C|}||w_k||_2^2$,  $\{w_k\}_{k=1}^{|C|}\subset\mathbb{C}$ being the  set of the codewords, $||\cdot||_2$ the Euclidean norm, and $N_0$ the noise variance.
\section{Arithmetic Fuchsian groups acting on $\mathcal{H}$}

\subsection{Quaternion algebras, orders and arithmetic Fuchsian groups} In this section, we survey some facts on the arithmetic of quaternion algebras. We mainly follow the references \cite{alsinabayer} and \cite{vigneras}.

Let $a,b\in\mathbb{Z}\setminus\{0\}$ and let
$H=\left(\frac{a,b}{\mathbb{Q}}\right)$
be the quaternion $\mathbb{Q}$-algebra generated by $I$ and $J$
with the standard relations
$I^2=a,J^2=b,IJ=-JI$.
Denote $K=IJ$. For $\omega= x+yI+zJ+tK\in H$, the conjugate is $\overline{\omega}=x-yI-zJ-tK$, and the reduced trace and the reduced norm are defined by
$$
\mathrm{Tr}(\omega)=\omega+\overline{\omega}=2x, \quad
\mathrm{N}(\omega)=\omega\overline{\omega}=x^2-ay^2-bz^2+abt.
$$
The following map is a monomorphism of $\mathbb{Q}$-algebras
$$
\begin{array}{ccc}
\phi: \left(\dfrac{a,\, b}{\mathbb{Q}}\right) & \to & \mathrm{M}_2\left(\mathbb{Q}(\sqrt{a})\right)\\
x+yI+zJ+tK & \mapsto &
\left(\begin{array}{ccc} x+y\sqrt{a} &\phantom{x} & z+t\sqrt{a}\\
b(z-t\sqrt{a})&\phantom{x} & x-y\sqrt{a}\end{array}\right).
\end{array}
$$
Notice that for any $\omega\in H$,
$\mathrm{N}(\omega)=\mathrm{det}\left(\phi(\omega)\right)$, and $\mathrm{Tr}(\omega)=\mathrm{Tr}\left(\phi(\omega)\right)$. In the rest of the paper, $H$ will denote a quaternion $\mathbb{Q}$-algebra.

For any absolute value $|\phantom{x}|_p$ of $\mathbb{Q}$ attached to a prime number $p$ or to $p=\infty$,
$H_p:=H\otimes_{\mathbb{Q}}\mathbb{Q}_p$ is a quaternion $\mathbb{Q}_p$-algebra.
If $H_p$ is a division algebra, it is said that $H$ is ramified at $p$.
As is well known, $H$ is ramified at a finite even number of places.
The discriminant $D_H$ is defined as the product of the primes at which $H$ ramifies.
Moreover, two quaternion $\mathbb{Q}$-algebras are isomorphic if and only if they have the same discriminant.

\begin{definition}
If $D_H=1$, $H$ is said to be non-ramified; in this case, it is isomorphic to $M_2(\mathbb{Q})$.
If $H$ is ramified at $p=\infty$, it is said to be definite, and indefinite otherwise.
An indefinite quaternion algebra is said to be small ramified if $D_H$ is equal to the product of two distinct primes.
\end{definition}

An element $\alpha\in H$ is said to be integral if $\mathrm{N}(\alpha),\mathrm{Tr}(\alpha)\in \mathbb{Z}$.
A $\mathbb{Z}$-lattice $\Lambda$ of $H$ is a finitely generated torsion free $\mathbb{Z}$-module contained in $H$.
A $\mathbb{Z}$-ideal of $H$ is a $\mathbb{Z}$-lattice $\Lambda$ such that $\mathbb{Q}\otimes \Lambda\simeq H$.
A $\mathbb{Z}$-ideal is not in general a ring.
An order $\mathcal{O}$ of $H$ is a $\mathbb{Z}$-ideal which is a ring.
Each order of a quaternion algebra is contained in a maximal order.
In an indefinite quaternion algebra, all the maximal orders are conjugate (cf.\,\cite{vigneras}).

\begin{definition}Two groups $G_1$ and $G_2$ are said to be commensurable if $G_1\cap G_2$ has finite index both in $G_1$ and in $G_2$.
\end{definition}

\begin{definition}
Given $D> 1$, fix a quaternion algebra $H$ of discriminant $D$. Let us denote by $\Gamma(D,1)$ the image under $\phi$ of the group of units of reduced norm $1$ in a maximal order of $H$. The groups $\Gamma(D,1)$ are Fuchsian groups of the first kind. A discrete group $\Gamma\subseteq \mathrm{GL}\left(2,\mathbb{R}\right)$ is said to be an arithmetic Fuchsian group of the first kind (Fuchsian group from now on) if there exists some quaternion $\mathbb{Q}$-algebra of discriminant $D$ such that $\Gamma$ is commensurable with $\Gamma(D,1)$.
\label{comentmodulareichler}
\end{definition}

\begin{remark}The notation $\Gamma(D,1)$ is a particular case of $\Gamma(D,N)$, which stands for the group of units of reduced norm 1 in a special kind of orders, so called Eichler orders, which are intersections of two maximal orders (cf. \cite{alsinabayer}).
\end{remark}

\subsection{Fundamental domains}

The group $\mathrm{SL}(2, \mathbb{R})$ acts on $\mathcal{H}$ by M\"{o}bius transformations
and its action factorizes through
$\mathrm{PSL}(2, \mathbb{R})$.

\begin{definition}Let $\Gamma$ be an arithmetic Fuchsian group commensurable with $\Gamma(D,1)$. A fundamental domain for the action of $\Gamma$ on $\mathcal{H}$ is a region $\mathcal{F}$ of $\mathcal{H}$ satisfying:
\begin{itemize}
\item[a)] For any $z,w\in\mathcal{F}$, if there exists $\gamma\in\Gamma$ such that $\gamma(z)=w$, then $z=w$ and $\gamma=Id$.
\item[b)] For any $z\in\mathcal{H}$, there exists $w\in\mathcal{F}$ and $\gamma\in\Gamma$ such that $\gamma(z)=w$.
\end{itemize}
\end{definition}
Fundamental domains for several groups $\Gamma(D,1)$ can be found in \cite{alsinabayer}.

\section{The point reduction algorithm: the case of an absolutely reliable channel}

Let us suppose that we have an alphabet consisting of a finite set of $4$-tuples of integers, say $A=\{(x_i,y_i,z_i,t_i)\}_{i=1}^N\subset\mathbb{Z}^4$. Let us also suppose that the elements $(x,y,z,t)\in A$ satisfy $x^2-ay^2-bz^2+abt^2=1$, the normic equation for $H$. We can think of these 4-tuples as elements of a hyper quadric in $\mathbb{R}^4$.

The way of sending $(x,y,z,t)\in A$ will be by sending $\gamma(\tau)$ where  $\gamma=$\\$\left(\begin{array}{cc}x+y\sqrt{a} & z+t\sqrt{a}\\b(z-t\sqrt{a}) & x-y\sqrt{a}\end{array}\right)$ and $\tau$ is an interior point of a prescribed fundamental domain $\mathcal{F}$ and with maximal distance to the boundary of $\mathcal{F}$. Notice that we are compressing $4$ information symbols, namely, $x,y,z$ and $t$ into $\gamma(\tau)\in \mathbb{C}$, and transmitting it by a single antenna. The usual way of sending 4 integers would be to use \emph{e.g.} spatial multiplexing and transmit 2 independent QAM symbols from two transmit antennas, or employ the channel two times sending one QAM symbol each time. Our method avoids involving more than one transmit antenna and employs the channel only once. In case the channel is very noisy at one channel use, there is still of course the option to use further coding over multiple antennas or over multiple channel uses by \emph{e.g.} a simple repetition.  But let us assume in this section that our channel is absolutely reliable. The problem is how to decode the received symbol $\gamma(\tau)$ to recover $(x,y,z,t)$, which is equivalent to obtain $\gamma$ out from $\gamma(\tau)$.

Since $\tau$ is an interior point, and interior points have trivial stabilizers, it suffices to find an algorithm which, given $z\in\mathcal{H}$, returns a representative $w\in\mathcal{F}$ of $z$. This problem is known as point reduction. Since $\pm\gamma$ induce the same map on $\mathcal{H}$, we would recover $\pm\gamma$. Hence, we will assume that we know how to decide if $\gamma$ or $-\gamma$ has been sent. We will prove that this is the case.

\subsection{Arithmetic Fuchsian groups of signature $(1;e)$}
We will develop our algorithm for the so called arithmetic Fuchsian groups $\Gamma$ of signature $(1;e)$ (for a precise definition cf. Takeuchi \cite{tak}). These groups $\Gamma$ are not the image under $\phi$ of a group $\Gamma(D,1)$, but indeed, the group $\Gamma^2$ generated by products of two elements of $\Gamma$ is. The index $[\Gamma:\Gamma^2]$ is explicitly computed according to a formula by Shimizu (cf. \cite{tak}). This being said, Takeuchi groups cannot be (at first glance) used to encode NUF symbols. The reason why we develop our algorithm for them, is that the general algorithm is more complicated. Nevertheless, we have implemented the algorithm also in the general case and we will plot our SNR/CER graphs with it. The authors are more than satisfied to provide an implementation of our algorithm to any interested reader.

For $g=\begin{pmatrix}a & b\\c & d\end{pmatrix}\in\Gamma$, if $g$ is not a homothety then denote by $I(g)$ the isometry circle of $g$, namely, the set $\{z\in\mathcal{H}\mid\quad\mid cz+d\mid =1\}$.
If $g$ is a homothety of factor $\lambda$,
then define $I(g)=\{z\in\mathcal{H}\mid\quad\mid\lambda z\mid =1\}$.
Denote by $\mathrm{Ext}(I(g))$ the exterior of $I(g)$ and by $\mathrm{Int}(I(g))$ the complement of $\mathrm{Ext}(I(g))$.
For any $\lambda\in\mathbb{R}$, $\lambda>0$, denote
\[
S(\lambda)=\{z\in\mathcal{H}\mid\,\,\lambda^{-1}\leq\mid z\mid\leq\lambda\}.
\]
If $h$ is a homothety of factor $\lambda$, then notice that the isometry circles of $h$ and $h^{-1}$ are parallel in the hyperbolic metric. Sometimes (cf.\,\cite{alsinabayer}), it is possible to find a system of generators $G$ of $\Gamma$ such that one of them is an hyperbolic homothety $h$ of factor $\lambda$ and a fundamental domain of the form
\[
\mathcal{F}=\bigcap_{g\in G\setminus\{h,h^{-1}\}}\mathrm{Ext}(I(g))\cap S(\lambda).
\]
In this case, we shall call $S(\lambda)$ the \emph{fundamental strip} of $\mathcal{F}$. We can construct such a fundamental domain, for instance, when $\Gamma$ is one of the 73 arithmetic Fuchsian groups of signature $(1;e)$, which were classified by Takeuchi in \cite{tak}.
These arithmetic Fuchsian groups admit a presentation of the form $\Gamma=\langle \alpha,\beta:\left(\alpha\beta\alpha^{-1}\beta^{-1}\right)^e=\pm 1\rangle$.

\begin{proposition}[Sijsling, \cite{sij}]
Let $\Gamma$ be an arithmetic Fuchsian group of signature $(1;e)$ generated by $\alpha$ and $\beta$. Then, after a change of variables, we can suppose that $\alpha$ is a homothety of factor $\lambda$ and $\beta=\begin{pmatrix}a & b\\b & a\end{pmatrix}$. Furthermore, the hyperbolic rectangle $\mathcal{F}=S(\lambda)\cap\mathrm{Ext}(I(\beta))\cap\mathrm{Ext}(I(\beta^{-1}))$ is a fundamental domain for $\Gamma$.
\label{generatorsTak}
\end{proposition}
Figure \ref{dominio fundamental} shows a fundamental domain for the group of signature $(1;2)$ labeled $e2d1D6ii$ in \cite{sij}. In this case,
\[
\alpha=\begin{pmatrix}\frac{\sqrt{6}}{2}+\frac{\sqrt{2}}{2} & 0\\0 & \frac{\sqrt{6}}{2}-\frac{\sqrt{2}}{2}\end{pmatrix}\text{ and } \beta=\begin{pmatrix}\sqrt{2} & 1\\1 & \sqrt{2}\end{pmatrix}.
\]

Let $\Gamma$ be an arithmetic Fuchsian group of signature $(1;e)$ generated by $\alpha,\beta$. Our aim is to develop an algorithm that given $z\in\mathcal{H}$ returns $w\in\mathcal{F}$ and $\gamma\in\Gamma$ such that $\gamma(z)=w$. This problem is equivalent to the following: given $g\in\Gamma$, to express $g$ as a product of powers of $\alpha$ and $\beta$. The idea is to multiply $g$ on the left by a suitable sequence of matrices $\{g_{k_j}\}$, with $g_{k_j}$ a power of $\alpha$ or $\beta$ to obtain a product $g_{k_1}\cdots g_{k_n}g$, such that $g_{k_1}\cdots g_{k_n}g(\tau)$ belongs to the interior of $\mathcal{F}$, being $\tau$ as above. In this case, $g=\left(g_{k_1}g_{k_2}\cdots g_{k_n}\right)^{-1}$. Observe that the decomposition of $g$ as a product of generators is not unique in general.


\begin{figure}
\centering
\scalebox{0.35}{\includegraphics{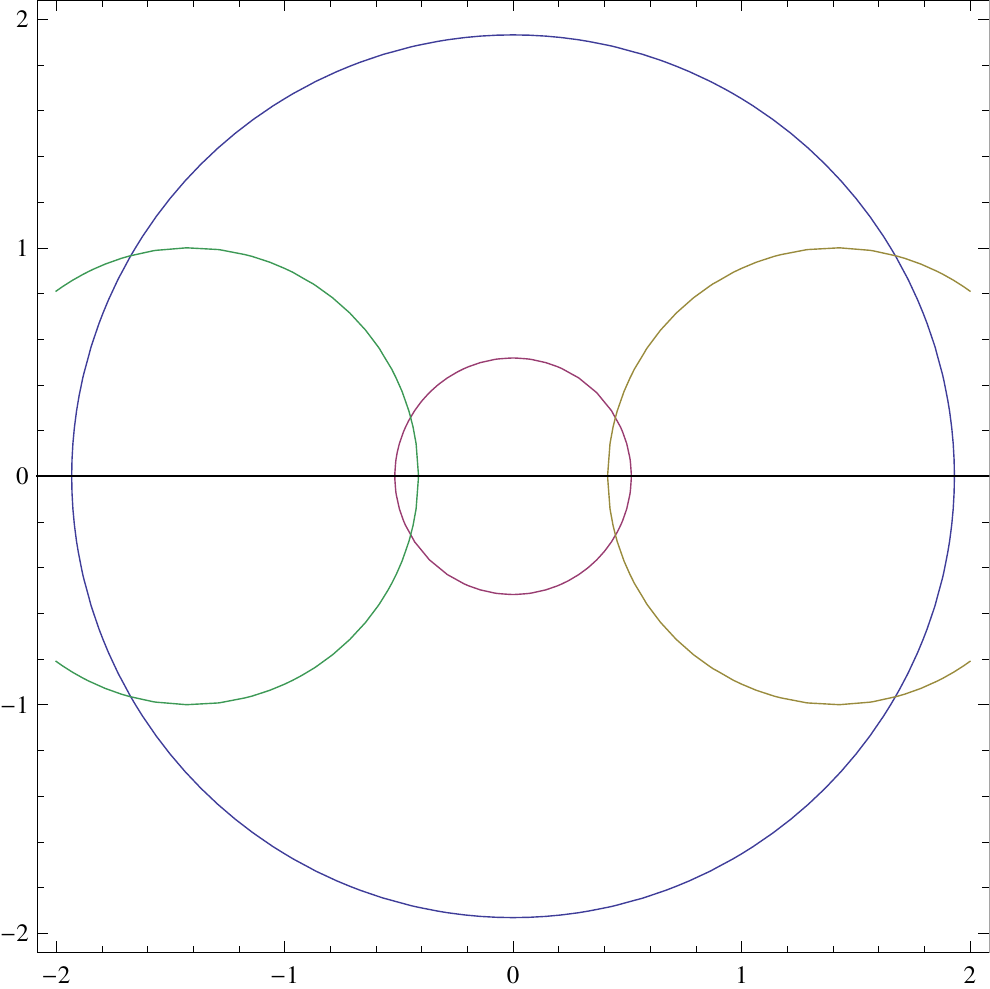}}\hspace{2cm}
\caption{Fundamental domain for $\Gamma=e2d1D6ii$.}
\label{dominio fundamental}
\end{figure}
\vskip 2mm
Let $z\in\mathcal{H}$ and $g\in\Gamma$ be such that $g(z)\not\in S(\lambda)$. Define $N(g)\in\mathbb{Z}$ such that $\lambda^{-1}\leq|\alpha^{N(g)}(\tau)|\leq\lambda$. We propose the following procedure, whose theoretical justification has been given in \cite{blancoboix}.


\begin{algorithm}\caption{Decomposition into distinguished closed paths}\label{hola1}
\begin{algorithmic}
\label{algoritmo1}
\begin{footnotesize}
\REQUIRE $g\in\Gamma,\tau\in\mathrm{Int}\left(\mathcal{F}\right)$.
\ENSURE $\{n_{\alpha},n_{\beta}\}$ such that $\{\tau,g(\tau)\}=n_{\alpha}\{\tau,\alpha(\tau)\}+n_{\beta}\{\tau,\beta(\tau)\}$.
\medskip
\STATE $\gamma\gets g,n_{\alpha}\gets 0,n_{\beta}\gets 0$;
\STATE $flag=false$.
\WHILE {$flag==false$}
\IF{$\gamma(\tau)\not\in S$}
\STATE $n_{\alpha}\gets n_{\alpha}+N(g)$.
\STATE $g\gets\alpha^{N(g)}g$.
\STATE $\gamma\gets\gamma\alpha^{-N(g)}$.
\ELSE
\IF{$\gamma(\tau)\in\mathcal{F}$}
\STATE $flag\gets true$.
\ENDIF
\ENDIF
\IF{$\gamma(\tau)\in S^{+}$}
\STATE $n_{\beta}\gets n_{\beta}+1$.
\STATE $g\gets\beta^{-1}g$.
\STATE $\gamma\gets\gamma\beta$.
\ENDIF
\IF{$\gamma(\tau)\in S^{-}$}
\STATE $n_{\beta}\gets n_{\beta}-1$.
\STATE $g\gets\beta g$.
\STATE $\gamma\gets\gamma\beta^{-1}$.
\ENDIF
\ENDWHILE
\RETURN $\{n_{\alpha},n_{\beta}\}$ such that $\{\tau,g(\tau)\}=n_{\alpha}\{\tau,\alpha(\tau)\}+n_{\beta}\{\tau,\beta(\tau)\}$.
\end{footnotesize}
\end{algorithmic}
\end{algorithm}



\subsection{Complexity} As promised, the properly discontinuous character of the action of the group $\Gamma$ of signature $(1;e)$ implies fast decodability. Let $\{(x_k,y_k,z_k,t_k)\}_{k=1}^{N}$ be the set of integral $4$-tuples to be encoded, and $\{\gamma_{k}\}_{k=1}^{N}$ the corresponding set of matrices. Since we have chosen $\tau$ in the interior of $\mathcal{F}$, $|C|=N$.
\begin{proposition}The total complexity of the point reduction algorithm (counting points, comparisons and additions) is at most $10\log_2\left(|C|+1\right)$.
\label{complexity}
\end{proposition}
\begin{proof}In the proof of the correctness of our algorithm, if a matrix is expressible as a product of $m$ generator matrices, the algorithm does exactly $m$ steps, each step consisting of at most $2$ comparison and $1$ matrix product (which accounts for at most $4$ products and $4$ sums of real numbers). Hence, $10m$ operations if the matrix is a product of $m$ generators.

On the other side, we can write $|C|\leq \sum_{k=0}^M2^k=2^{M+1}-1$ for certain $M\geq 1$. Indeed, we can decompose $|C|$ as a sum in which the $k$-th is at most the number of matrices which are a product of $k$ matrices. Hence, the worst case happens when the algorithm needs to decompose a matrix which is product of $M=\log_2(|C|+1)$ generators. It performs $10M$ operations.
\end{proof}

\section{The Gaussian channel}

Now, we suppose that our channel is affected by an additive white Gaussian noise, which we model as a sequence of independent identically distributed random variables $CN(0,\Sigma)$ with $||\Sigma||=N_0$. In this case, we have $N$ 4-tuples encoded by matrices $\{\gamma_1,...,\gamma_N\}\in\Gamma$. Fix a fundamental domain $\mathcal{F}$ for $\Gamma$ and let $\tau\in\mathcal{F}$ be an interior point with maximal distance to the boundary of $\mathcal{F}$. As in the previous section, if we want to transmit the $4$-tuple $(x_k,y_k,z_k,t_k)$, we actually send the corresponding matrix acting on $\tau$, \emph{i.e.}, $\gamma_k(\tau)$. But the receiver will have $\gamma_k(\tau)+n_k$, where $n_k$ is a realization of $CN(0,\Sigma)$ at the time in which the receiver receives the information. The fact that we have as many matrices as $4$-tuples is due to the fact that the point $\tau\in\mathcal{H}$ on which our matrices act is chosen to be interior to the prescribed fundamental domain $\mathcal{F}$.

\begin{lemma}Let $\Gamma$ be a Fuchsian group. If $\gamma_1(\mathcal{F})=\gamma_2(\mathcal{F})$ for some $\gamma_1,\gamma_2\in\Gamma$, then $\gamma_1=\pm\gamma_2$.
\end{lemma}
\begin{proof}Take an interior point $\tau\in\mathcal{F}$. Then, $\gamma_1(\tau)=\gamma_2(w)$, where $w\in\mathcal{F}$ is another interior point. Hence, $w=\gamma_2^{-1}\gamma_1(\tau)$. Since interior points cannot be congruent each other modulo $\Gamma$, it follows that $w=\tau$ and $\gamma_1=\pm\gamma_2$.
\end{proof}

Hence, suppose that $n_k$ is small enough to ensure that both $\gamma_k(\tau),\gamma_k(\tau)+n_k\in\gamma_k(\mathcal{F})$. In this case, according with the previous lemma, the reduction algorithm will return the same matrix (up to a sign).

\subsection{An alternative Fuchsian group} Here, we explore the Fuchsian code attached to $\Gamma(6,1)$ (see notations in 2.2.1).

\begin{theorem}[Alsina, Bayer,\, cf. \cite{alsinabayer}]Given $\alpha=a+b\sqrt{3}\in\mathbb{Z}[\sqrt{3}]$ denote $\alpha'=a-b\sqrt{3}$. Then, the group $\Gamma(6,1)$ is
$$
\left\{\gamma=\dfrac{1}{2}\left(\begin{array}{cc}\alpha & \beta\\-\beta' & \alpha'\end{array}\right)\mid\alpha,\beta\in\mathbb{Z}[\sqrt{3}],\mathrm{det}(\gamma)=1,\alpha\equiv\beta\equiv\alpha\sqrt{3}\pmod{2}\right\}.
$$
\end{theorem}

This group contains $\Gamma^2$, the group generated by the squares of the elements of the Takeuchi group $\Gamma$ labeled $e2d1D6ii$ in \cite{sij}. Furthermore, $\Gamma(6,1)$ is the image by $\phi$ of the multiplicative group of elements of reduced norm $1$ in a maximal order (unique up to conjugation) of $\left(\frac{3,-1}{\mathbb{Q}}\right)$. Notice that we are interested in a particuar subgroup of $\Gamma(6,1)$, namely, the subgroup of matrices whose entries belong to $\mathbb{Z}[\sqrt{3}]$, since our codewords consists of $4$-tuples of integers. This is the image by the left-regular representation of the group of unit of reduced norm $1$ in the natural order of the quaternion $\mathbb{Q}$. Hence, if $(x,y,z,t)$ is a codeword encoded by an element of $\Gamma(6,1)$, it satisfies $x^2-3y^2+z^2-3t^2=1$.

\begin{lemma}The sum of the elements of the first (or second) column of any element of $\Gamma(6,1)$ is non zero.
\end{lemma}
\begin{proof}If this were not the case, let $\gamma=\frac{1}{2}\left(\begin{array}{cc}\alpha & \beta\\-\beta' & \alpha'\end{array}\right)\in\Gamma(6,1)$ with $\alpha-\beta'=0$. We would have that $\beta\beta'=2$, in particular we would have that $2$ is a square modulo $3$.
\label{sumafilas}
\end{proof}

\begin{remark}
We point out that the reduction algorithm can recover the transmitted matrix up to a sign; how to determine the sign will be discussed in the further research section.
\end{remark}





\begin{theorem}[\cite{bayerremon}]\label{redpointalg}
There exists a reduction point algorithm for each cocompact Fuchsian group. Fixed $C$ a codebook and $n = |C|$, the complexity is asymptotically  logarithmic in $n$.
\end{theorem}

\begin{proof}
The proof uses isometric circles in order to determine which map has to be applied in each step of the reduction algorithm.

The fundamental domains in these particular cases are given by a number of isometric circles equal twice the number of generators of $\Gamma$. Therefore, in each step of the reduction algorithm, we perform at most twice the number of generators comparisons and $2\times 2$ matrix products.

The complexity of the algorithm is due to the fact that once we fixed a code we will know in advance the number of the steps the algorithm is going to use. With a fixed number of steps,  say $n$, we have, asymptotically $e^n$ codes which reduce with $n$ steps.
\end{proof}


\begin{remark}
Recall that linearly structured codes are bound to have a worst-case maximum-likelihood (ML) decoding complexity proportional to $|S|^{\kappa}=|C|$, where $|S|$ is the size of the  underlying (\emph{e.g.}, QAM) symbol alphabet, $\kappa$ is the number of independent (QAM) symbols in one codeword, and $|C|$ is the size of the resulting code. The complexity is measured in the number of metric evaluations $||y-hx||_F$ one has to do in order to decode. Notice that each metric evaluation involves two multiplications of real numbers and one addition. So when measuring the complexity in number of operations, the complexity becomes $3|C|$ in the SISO case, and $3n_tT|C|$ in the $n_t\times n_r$ MIMO case, where $n_t,n_r$ and $T$ denote the number of transmit and receive antennas and decoding delay (=matrix block length), respectively.

Usually linear codes are preferred due to simple decoding methods such as a sphere decoder, but here we have seen that by using nonlinear Fuchsian codes one can actually reduce the decoding complexity from $O(|C|)$ to $O(\log |C|).$

The constant terms are nevertheless relevant when the code size is small; it can be shown in an analogous way that in the proof of Proposition \ref{complexity}, that the Fuchsian code $\Gamma(6,1)$ has complexity $19\log_3(2|C|+1)$ which is less than $3|C|$ for $|C|\geq 30$.
\end{remark}

We have to be careful with the choice of $\tau$. If $\gamma(\tau)$ is close to the boundary of $\gamma(\mathcal{F})$, it will be \emph{less immune} to noise than other point further to the boundary: noise in the direction of minimal distance with this minimal distance as magnitude could eventually bring the point outside the fundamental domain, while this is less likely to happen to a point of bigger distance to the border. We have adopted the following optimality criterion for choosing $\tau$: denoting by $\partial\gamma_k(\mathcal{F})$ the boundary of $\gamma_k(\mathcal{F})$ and by $\mathrm{Int}\left(\gamma_k\left(\mathcal{F}\right)\right)$ its interior, let $c_k\in\gamma_k(\mathcal{F})$ be such that $d(c_k,\partial\gamma_k(\mathcal{F}))=\displaystyle\max_{z\in \gamma_k(\mathcal{F})} d(z,\partial\gamma_k(\mathcal{F}))$. Then, we have to find $z\in\mathrm{Int}(\mathcal{F})$ minimizing $\sum_{k}|\gamma_k(z)-c_k|^2$, \emph{i.e.}, the point $\tau$ such that its transforms are simultaneously closest to the corresponding centers on quadratic average. The choice of such an optimal $\tau$ is critical to the performance of the code, as also confirmed by our earlier simulations. Notice that different quaternion algebras can also yield drastically different performances.

\subsection{Generating the constellation}

We address now the problem of how to produce the $4$-tuples $(x,y,z,t)\in\mathbb{Z}^4$ such that $x^2-ay^2-bz^2+abt^2=1$, which will be sent in the form of matrices of $\Gamma(D,1)$ acting on $\tau$ by M\"{o}bius transforms. We will restrict ourselves to $a=3$ and $b=-1$, to derive explicit results, but the same analysis works in general.

As we have said before, only three symbols in each $4$-tuple are independent, hence, we would like to parametrize the set of these $4$-tuples by an infinite set of $3$-tuples $(m,k_1,k_2)\in\mathbb{Z}$. Since the quaternion algebra $\left(\frac{3,-1}{\mathbb{Q}}\right)$ is indefinite, one has that the normic equation $x^2-3y^2+z^2-3t^2=1$ has infinitely many integer solutions (cf.\cite{alsinabayer}). It is possible to parametrize all the rational solutions of this normic equation by means of rational functions in three variables, but using this method to produce integer solutions seems a difficult task. We develop instead, an alternative method to produce an infinite set of such solutions. Next, we detail our construction.

First, notice that the ring of integers of the number field $\mathbb{Q}(\sqrt{3})$ is $\mathbb{Z}[\sqrt{3}]$. The multiplicative group of units of this ring is $\left\{\pm\varepsilon^m:m\in\mathbb{Z}\right\}$, where $\varepsilon$ is a unit of infinite order (called a fundamental unit). This is a very particular version of Dirichlet's theorem on units, but in our case it is easy to find by inspection such a fundamental unit. We will consider $\varepsilon=2+\sqrt{3}$. Given an element $\theta=a+\sqrt{3}b\in \mathbb{Q}(\sqrt{3})$, let us denote by $\theta'$ its Galois conjugated, \emph{i.e.}, $a-\sqrt{3}b$.

Given $(m,k_1,k_2)$ a triple of non-negative integers ($m\neq 0$), define $a_m+\sqrt{3}b_m=\varepsilon^m$. We have that $a_m^2-3b_m^2=\varepsilon^m(\varepsilon')^m=1$. Now, set $x_{m,k_1}+\sqrt{3}y_{m,k_1}:=a_m\varepsilon^{k_1}$ and $z_{m,k_2}+\sqrt{3}t_{m,k_2}:=\sqrt{3}b_m\varepsilon^{k_2}$. Notice that $x_{m,k_1}^2-3y_{m,k_1}^2=a_m^2$ and $z_{m,k_2}^2+3t_{m,k_2}^2=-3b_m^2$, hence
$$x_{m,k_1}^2-3y_{m,k_1}^2+z_{m,k_2}^2-3t_{m,k_2}^2=a_m^2-3b_m^2=1.$$
We will denote by $\phi(m,k_1,k_2)$ the $4$-tuple $(x_{m,k_1},y_{m,k_1},z_{m,k_2},t_{m,k_2})$. This way, we have parametrized by three variables an infinite subset of integer points of the hyper quadric $x^2-3y^2+z^2-3t^2=1$.

\begin{proposition}The map $\phi$ is bijective over its image, which is contained in the set $\{(x,y,z,t)\in\mathbb{Z}_{\geq 0}^4:x^2-3y^2=m^2, z^2-3t^2=-3r^2,\mbox{ for some }m,r\in\mathbb{Z}\}$.
\end{proposition}
\begin{proof}Let $(m_1,k_{1,1},k_{1,2})$ and $(m_2,k_{2,1},k_{2,2})$ two triples of non-negative integers with $m_1,m_2\neq 0$. Suppose $m_1=m_2=m$. If $k_{1,1}\neq k_{2,1}$, then $a_m\varepsilon^{k_{1,1}}\neq a_m\varepsilon^{k_{2,1}}$ and $\phi(m_1,k_{1,1},k_{1,2})\neq \phi(m_2,k_{2,1},k_{2,2})$. The case  $k_{2,1}\neq k_{2,2}$ is analogous. Suppose that $m_1\neq m_2$. In this case, $a_{m_1}\neq a_{m_2}$ or $b_{m_1}\neq b_{m_2}$. Suppose that $a_{m_1}\neq a_{m_2}$. In this case, $a_{m_1}\varepsilon^{k_{1,1}}\neq a_{m_2}\varepsilon^{k_{2,1}}$, since otherwise, $a_{m_1}^2=a_{m_2}^2$, and since $\varepsilon>0$, we would have that $a_{m_1}=a_{m_2}$. The remaining case is identical.
\end{proof}

As an example, we have that $\phi(1,0,1)=(2,0,3,2)$, $\phi(2,0,1)=(7,0,12,8)$, or $\phi(2,1,1)=(14,7,12,8)$. Notice that the last two values are difficult to obtain merely by inspection.

The explicit parametrization of the whole group of units is a delicate problem. On the contrary to the number field setting, the structure of the group of units of reduced norm $1$ in quaternion algebras has not been explicitly described yet. However, there exist some interesting theoretical results, as \cite{capi}.

\subsection{An alternative approach}

In this subsection we fix $D=6$, $N=1$, and the quaternion algebra of discriminant $6$, $H=\left(\frac{3,-1}{\mathbb{Q}}\right)$; but analogous results  hold in general. From the theory of embeddings of quadratic fields into quaternion algebras (cf.\cite{alsinabayer}) it is known that there exist embeddings of $\mathbb{Q}(\sqrt{d})$ into $H$ for any square-free $d>0$ such that $\left(\frac{d}{2}\right)=\left(\frac{d}{3}\right)=-1$ (for example $d=3,-1,6$).

Fixing an embedding of $\mathbb{Q}(\sqrt{d})$ into $H$ is equivalent to fix  a pure quaternion $\omega=xI+yJ+zK\in H$ of norm $-d$, that is $(x,y,z)\in\mathbb{Z}^3$ such that $3x^2-y^2+3z^2=d$. Since $H$ is indefinite, this normic equation has infinitely many solutions, hence, there exist bijections $\varphi_d:\mathbb{N}\to \{xI+yJ+zK\in\mathbb{Z}^3:3x^2-y^2+3z^2=d\}$. Determining such a bijection is equivalent to solve the diophantine equation $3x^2-y^2+3z^2=d$, which is a classical problem in Number Theory. It is possible to give asymptotic estimates of the number of solutions, which involves the use of modular forms of fractional weight $3/2$ (cf. \cite{duke}). Nevertheless, there exists a polynomial algorithm which computes finite sets of solutions (cf. \cite{simon}).

Now, given a real quadratic field $\mathbb{Q}(\sqrt{d})$ embedded in $H$ we can obtain units in $H$ from the group of units of the quadratic field, generated by $\varepsilon=a+b\sqrt{d}$  (notice that the fundamental unit $\varepsilon$ is usually normalized so that $a,b>0$ and
its absolute value is greater than $1$). Thus, identifying the units in the quaternion order with the corresponding matrices in the arithmetic Fuchsian group $\Gamma(6,1)$, we define maps $\psi_d:\mathbb{N}^2\to \Gamma(6,1)$ given by $\psi_d(t,m)=\left(a+b\varphi_d(t)\right)^m$.

\begin{proposition}The map $\psi_d$ is injective when restricted to $\mathbb{N}\times\left(\mathbb{N}\setminus\{0\}\right)$.
\end{proposition}

\begin{proof}
Consider $\psi_d(t_1,m_1)=\psi_d(t_2,m_2)$. Suppose first that $m_1=m_2=m$. Then, since we have $\varepsilon^m=l+r\sqrt{d}$, with $r\neq 0$, from $
l+r\varphi_d(t_1)=l+r\varphi_d(t_2)
$, we deduce  $\varphi_d(t_1)=\varphi_d(t_2)$, hence $t_1=t_2$.

Suppose now that $m_1\neq m_2$. In this case, setting $\psi_d(t_1,n_1)=l_1+m_1\varphi_d(t_1)$ and $\psi_d(t_2,m_2)=l_2+r_2\varphi_d(t_2)$, since $r_i\varphi_d(t_i)$ is a pure quaternion, we have that $l_1=l_2$.
However, by the binomial formula, and taking into account that $\varphi_d(t_1)^2=\varphi_d(t_2)^2=d$, we have
$$
l_1=\sum_{\tiny\begin{array}{l}j=0 \\ j \text{ even} \end{array}}^{m_1}{{m_1}\choose{j}}b^{j}d^{\frac{j}{2}}a^{m_1-j}.$$
Now assume  $m_1>m_2\geq j$, so we have that ${{m_1}\choose{j}}>{{m_2}\choose{j}}$, hence
$$ l_1 > \sum_{\tiny\begin{array}{l}j=0 \\ j \text{ even} \end{array}}^{m_2}{{m_1}\choose{j}}b^{j}d^{\frac{j}{2}}a^{m_1-j}
>\sum_{\tiny\begin{array}{l}j=0 \\ j \text{ even} \end{array}}^{m_2}{{m_2}\choose{j}}b^{j}d^{\frac{j}{2}}a^{m_2-j}=l_2,
$$
which is a contradiction.
\end{proof}

With these maps we can produce a countable family of non-overlapping infinite families of codewords:
\begin{proposition}Let $d_1,d_2$ be two different square-free positive integers such that
$\mathbb{Q}(\sqrt{d_1}),\mathbb{Q}(\sqrt{d_2})\hookrightarrow H$. Then
$\psi_{d_1}(t_1,m_1)=\psi_{d_2}(t_2,m_2)$ if and only if $m_1=m_2=0$.
\end{proposition}
\begin{proof}The \emph{if} clause is trivial.
Suppose $\psi_{d_1}(t_1,m_1)=\psi_{d_2}(t_2,m_2)$. Writing $\psi_{d_1}(t_1,m_1)=l_1+r_1\varphi_{d_1}(t_1)$ and $\psi_{d_2}(t_2,m_2)=l_2+r_2\varphi_{d_2}(t_2)$, we have that $l_1=l_2$ and $r_1\varphi_{d_1}(t_1)=r_2\varphi_{d_2}(t_2)$. Taking squares we obtain $r_1^2d_1=r_2^2d_2$, which implies $r_1=r_2=0$ and, since $d_1,d_2>0$, we deduce that $m_1=m_2=0$.
\end{proof}

The above facts, allow us to conclude the following

\begin{theorem}Let $\Gamma$ be the subgroup of $\Gamma(6,1)$ consisting of matrices with entries in $\mathbb{Z}[\sqrt{3}]$. There exists a parametrization of an infinite subset of $\Gamma$  by three degrees of freedom.
\end{theorem}
\begin{proof}Let $A$ be the infinite set of square-free integers $d$ such that $\mathbb{Q}(\sqrt{d})$ embeds into $\left(\frac{3,-1}{\mathbb{Q}}\right)$. For any $d>0$, fix a generator of the unit group of the form $a_d+b_d\sqrt{d}$ with $a,b>0$. Now, the map $\Psi:A\times\mathbb{N}\times\left(\mathbb{N}\setminus\{0\}\right)\to\Gamma(6,1)$ defined by $\Psi(d,s,m)=\psi_d(s,m)$ is injective. 
\end{proof}

\begin{remark}Notice that this theorem is not explicit, since it depends on how to produce the solutions of the normic form. But using the algorithm described in \cite{simon}, we can explicitly parametrize an infinite family of units by two dregrees of freedom. Further studies on the structure of the group of units will allow us to make the full parametrization more explicit.
\end{remark}

\subsection{Duplicating the size}

As a last step in our design, we can duplicate the size of the codebook in the following way: once we have made an optimal choice of matrices of $\Gamma(D,N)$ and $\tau\in\mathcal{H}$ having a codebook $C=\{\gamma_k(\tau)\}_{k=1}^{|C|}$, we can consider the new codebook $C=\{\pm\gamma_k(\tau)\}_{k=1}^{|C|}$. If a matrix $\gamma$ corresponds with the $4$-tuple $(x,y,z,t)$, and this $4$-tuple corresponds to the $3$-tuple $(m,k_1,k_2)$ of independent non-negative integers, we can impose that the matrix $-\gamma$ corresponds to the $3$-tuple $(-m,k_1,k_2)$. Notice that this is not ambiguous since the original triples are assumed to have non negative entries, and $\theta>0$.

To recover the right $3$-tuple from a received signal, first, we check if it belongs to $\mathcal{H}$ or to $-\mathcal{H}$. In the first case, we use the point reduction algorithm to obtain $(x,y,z,t)$ and the parametrization to obtain $(m,k_1,k_2)$. In the second case, we have received (unless the channel is in outage and an error is unavoidable) $v=-\gamma_k(\tau)+n$, hence, we apply the point reduction algorithm to $-v$, obtain $(x,y,z,t)$ and $(m,k_1,k_2)$, and we decode it as $(-m,k_1,k_2)$. 

    \begin{figure}
        \begin{center}
                \scalebox{1}{
                        \includegraphics[width=0.45\textwidth]{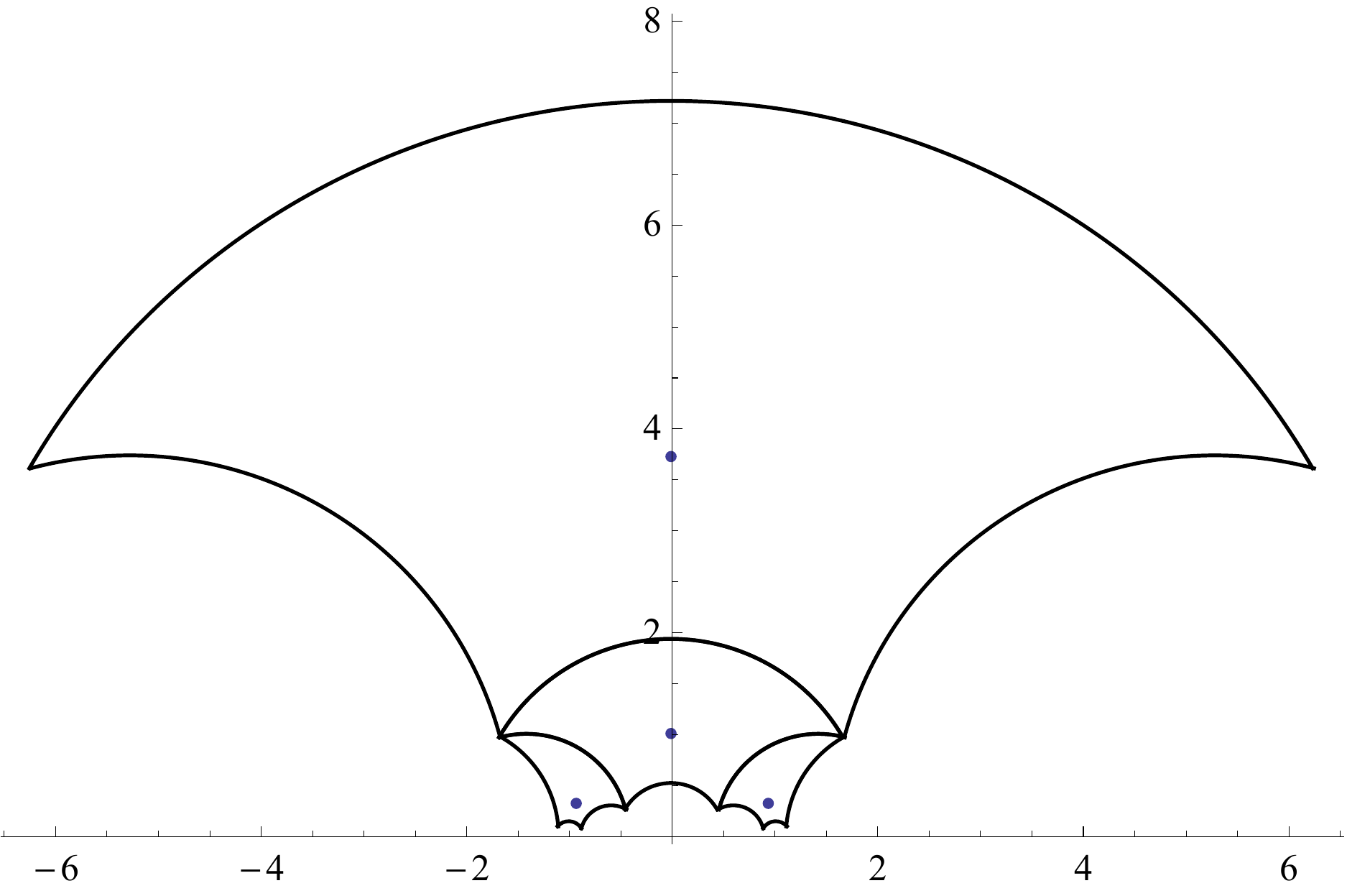}
                }
                        \caption[]{Example of constellation for $\Gamma=e2d1D6ii$}\label{constellation}
        \end{center}
    \end{figure}

\section{Simulation results}

    \begin{figure}
        \begin{center}
                \scalebox{1}{
                        \includegraphics[width=0.7\textwidth]{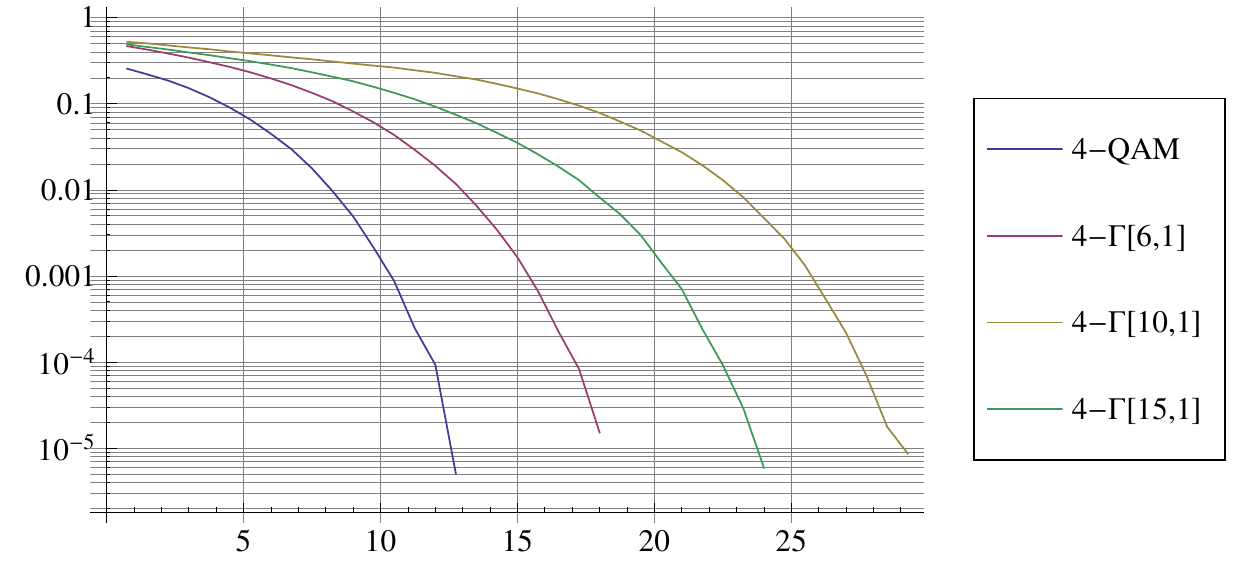}
                }
                        \caption[]{$4$ NUF codes vs $4$ QAM}\label{codes4}
        \end{center}
        \begin{center}
                \scalebox{1}{
                        \includegraphics[width=0.7\textwidth]{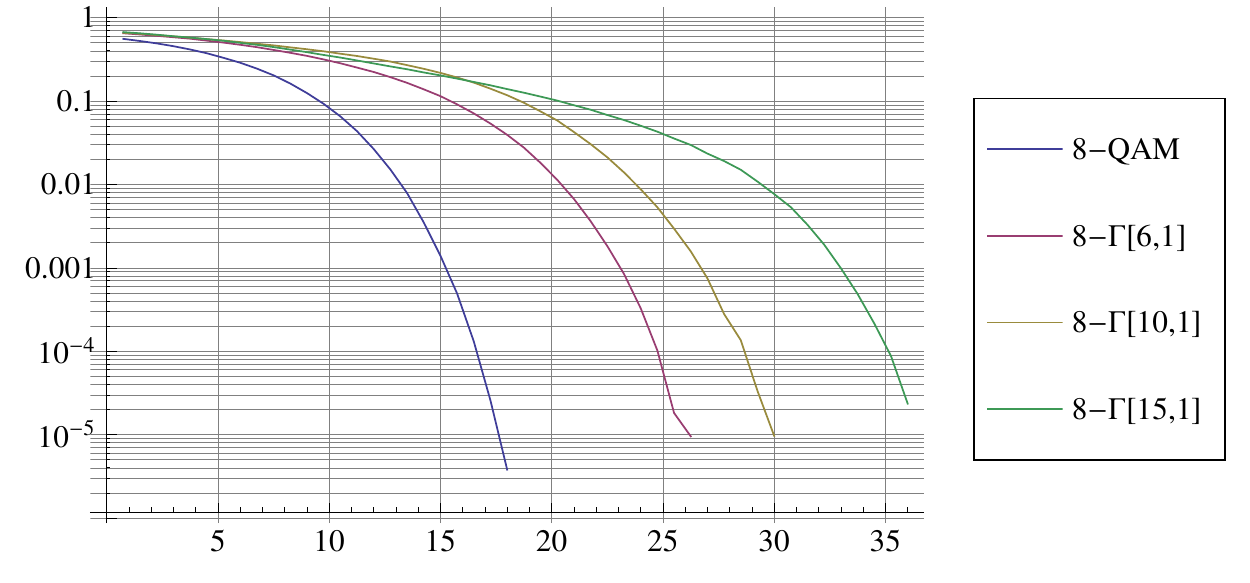}
                }
                        \caption[]{$8$ NUF codes vs $8$ QAM}\label{codes8}
        \end{center}
        \begin{center}
                \scalebox{1}{
                        \includegraphics[width=0.7\textwidth]{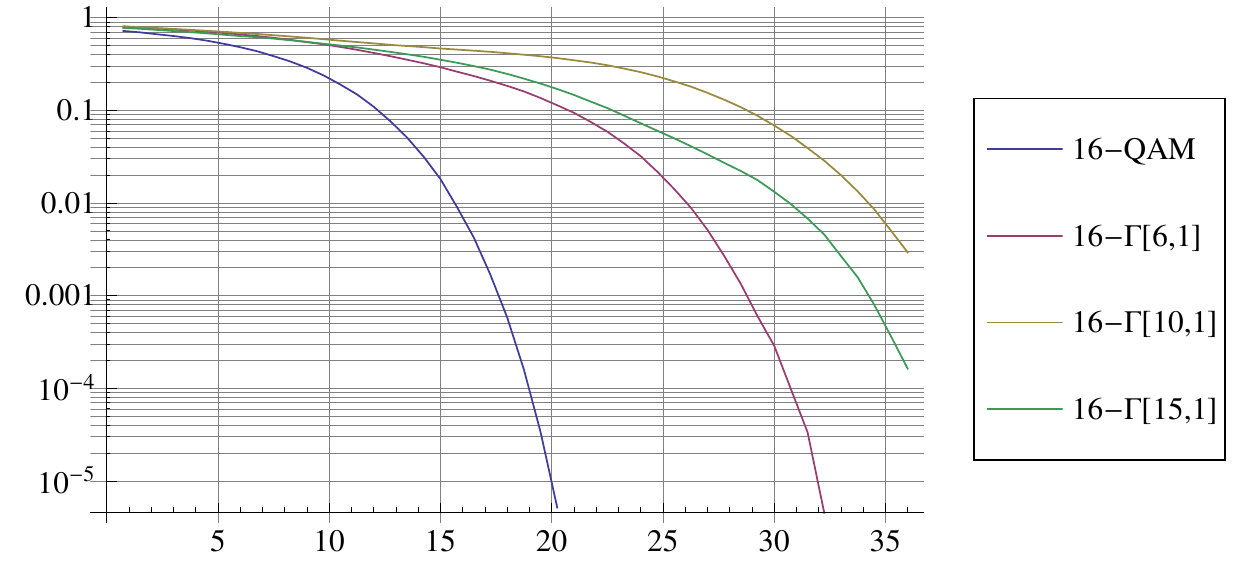}
                }
                        \caption[]{$16$ NUF codes vs $16$ QAM}\label{codes16}
        \end{center}
    \end{figure}

The simulations displayed in Figures \ref{codes4}, \ref{codes8} and \ref{codes16} have been done with $10^6$ rounds to compute the relative frequency of errors at each level of SNR. By $8$-QAM we mean a subset of $8$ symmetric symbols of a $16$-QAM constellation having minimal average energy. We can observe that the best of our $4$-NUF codes is still clearly outperformed by the $4$-QAM except for very low SNRs, but the gap to the worse 4-NUF is so vast that it gives hope to improve by another similar gap, which would bring us very close to 4-QAM. Considering the logarithmic decoding complexity\footnote{In general, there are also fast decoding algorithms for QAM constellations, but they come with a complexity--performance tradeoff, meaning that also the performance of a QAM constellation is degraded if we use suboptimal algorithms that are faster.}, some performance loss can easily be tolerated.  For bigger constellations the gap to QAM seems to grow, but luckily there are many more groups, fundamental domains, tessellations, generators, and centers $\tau$ to be tested that might have significantly better performance, as already observed in these preliminary simulations. In addition, our tentative results suggest that the point reduction algorithm can be improved at the penalty of increasing the complexity to $O(\log^2|C|)$.

\begin{remark} Non uniform constellations have been used in early-state signal transmission, the so-called  \emph{codec} transmission. Nowadays, the usage of certain non uniform constellations like ConQAM or NUQAM is being discussed and seriously considered for the design of MIMO systems for digital video (terrestrial) broadcasting (DVB-T2). While at this point we are only considering the AWGN channel, our aim is to generalize this framework for wireless fading channels and to design codes directly applicable to the DVB framework (for more information, see \cite{dvb}).
\end{remark}

\section{Further research}

As we have seen, the performance gap for different groups $\Gamma$, as $\Gamma=e2d1D6ii$, $\Gamma(6,1)$, $\Gamma(10,1)$ can be remarkable. In future work we will compare the performance and complexity for a bigger number of $\Gamma(D,1)$ and for different codebook sizes, as well as the explicit parametrization of all the $4$-tuples. A further study on our Fuchsian codes should also include the issue of error correction after point reduction, while not substantially increasing the complexity.


\subsection{Towards the fading channel}

For now, our scheme is valid only for AWGN channels, since in a fading channel a message $\gamma_k(\tau)$ would reach the receiving end as $h\gamma_k(\tau)+n$, with $h$ a random variable with circular Gaussian distribution $CN(0,1)$. One common simplification in the existing literature is to suppose that the receiver has perfect knowledge of the fading coefficient $h$ (via sending pilots in the signal, for instance).

It is well known that we can write $h=re^{i\theta}$ where $r$ is Rayleigh distributed and $\theta$ is uniformly distributed in $[0,2\pi]$. It is also common to suppose that the noise $n$ originates at the receiver side, hence, if we suppose as a first approach that $h=e^{i\theta}$, we can still use our scheme to decode: since $n$ is distributed as $CN(0,\Sigma)$, once the signal $v=e^{i\theta}\gamma_k(\tau)+n$ reaches the receiver, we perform $v^{*}=e^{-i\theta}y=\gamma_k(\tau)+e^{-i\theta}n$. Obviously, $e^{-i\theta}n$ has the same distribution as $n$, since $n$ is a complex circular symmetric random variable. We further have $|n|=|e^{-i\theta}n|$, \emph{i.e.}, the noise does not get amplified. Hence, it is enough to suppose that $h$ takes values in $\mathbb{R}$ and is Rayleigh distributed.



\begin{thebibliography}{99}
%
%

\bibitem{alsinabayer} Alsina, M.; Bayer, P.: \emph{Quaternion orders, quadratic forms and Shimura curves.}.
CRM Monograph Series, 22. American Mathematical Society, Providence, RI, 2004. xvi+196 pp. ISBN: 0-8218-3359-6.

\bibitem{blancoboix} Bayer, P.; Blanco-Chac\'{o}n, I.; Boix, A.F: Cohomological interpretation of the modular symbol. Submitted. Available at arxiv: http://arxiv.org/abs/1204.5670

\bibitem{bayerremon} Bayer, P.; Rem\'{o}n, D.: A reduction point por cocompat Fuchsian groups. (In process).

\bibitem{viterbo2} Belfiore, J.C.; Oggier, F.; Viterbo, E.: \emph{Cyclic Division Algebras: A Tool for Space\"{\i}?`�Time Coding.} Foundations and Trends in
Communications and Information Theory, 4, n. 1 (2007) 1--95.

\bibitem{angeles} Wireless Communications and Networking Project at Autonomous University of Barcelona: http://gent.uab.cat/mavazquez-castro/

\bibitem{capi} Corrales, C.; Jespers, E.; Leal, G.; del R\'{i}o, A.: Presentations of the unit group of an order in a non-split quaternion algebra. \emph{Advances in Mathematics}, 186 (2004), 498--524.

\bibitem{brasil1} Carvalho, E.D; Andrade, A.A.; Palazzo, R. Jr., Vieira Filho, J.: Arithmetic Fuchsian groups and space-time block codes. Comput. Appl. Math. 30 (2011), no. 3, 485--498.

\bibitem{duke} Duke, W.; Schulze-Pillot, R.: Representation of integers by positive ternary quadratic forms and equidistribution of lattice points on ellipsoids. \emph{Invent. Math.}, 99 (1990), 49--57.

\bibitem{dvb} DVB Project, The Global Standard for Digital Television, \emph{dvb.org}.

\bibitem{elkies} Elkies, N.: Excellent codes from modular curves. \emph{Proceedings of the thirty-third annual ACM symposium on Theory of computing,} (2001) 200 -- 208 .


\bibitem{hollanti4} Hollanti, C.; Lahtonen, J.; Lu, H.-F.: Maximal Orders in the Design of Dense Space-Time Lattice Codes. \emph{IEEE Transactions on Information Theory}, 54, n. 10 (2008).



\bibitem{viterbo1} Oggier, F.; Viterbo, E.:\emph{Algebraic Number Theory and Code Design for Rayleigh Fading Channels.} Foundations and Trends in
Communications and Information Theory, 1, n. 3 (2004) 333--415.


\bibitem{shimura1967} Shimura, G.: Construction of class fields and zeta functions of algebraic curves. \emph{Ann. of Math. }(2) 85 (1967), 58--159.


\bibitem{sij} Sijsling, J.: \emph{Equations for arithmetic pointed tori}. Ph.D. Thesis, Universiteit Utrecht, 2010. Available at {\tt http://sites.google.com/site/sijsling/research}

\bibitem{simon} Simon, D.: Solving quadratic equations using unimodular quadratic forms. Math. Comp. 74 (2005), no. 251, 1531--1543.


\bibitem{tak} Takeuchi, K.: \emph{Arithmetic Fuchsian groups with signature $(1,e)$}. J. Math. Soc. Japan 35, No.\,3 (1983), 381-407.

\bibitem{hollanti3} Vehkalahti, R.; Hollanti, C.; Lahtonen, J.; Ranto, K.: On the Densest MIMO Lattices From Cyclic Division Algebras. \emph{IEEE Transactions on Information Theory}, 55, n. 8 (2009).

\bibitem{vigneras} Vign\'{e}ras, M.\,F.: \emph{Arithm\'{e}tique des alg\`{e}bres de quaternions}.
Lecture Notes in Mathematics 800. Springer, 1980. vii+169 pp. ISBN: 3-540-09983-2.

\end{thebibliography}


\end{document}